\documentclass[10pt, conference]{IEEEtran}

\usepackage[english]{babel}

\usepackage{xcolor}
\definecolor{newcolor}{rgb}{.8,.349,.1}

\definecolor{brown}{cmyk}{0,0.81,1,0.60}
\definecolor{magenta}{rgb}{0.4,0.7,0}
\definecolor{gray}{rgb}{0.5,0.5,0.5}
\definecolor{red}{rgb}{1,0,0}
\definecolor{green}{rgb}{0,0.5,0}
\definecolor{blue}{rgb}{0,0,1}

\usepackage{cite}
\usepackage{amsmath,amssymb,amsfonts}
\usepackage{graphicx, tikz}
\usetikzlibrary{quantikz2, quotes, fit}
\usepackage[caption=false,font=footnotesize]{subfig}
\usepackage{pgfplots}
\pgfplotsset{compat=1.18}

\usepackage[colorlinks=true, allcolors=blue]{hyperref}
\usepackage{orcidlink}

\newcommand{\adv}{\mathcal{A}}
\newcommand{\cpanet}{\texttt{IND-CPA-NOISE($\mathcal{E}$)}}
\newcommand{\experiment}{\texttt{qPKEnc}}

\newcommand{\dthree}{
    \begin{tikzpicture}[semithick]
        \begin{yquant}[register/minimum height=3mm, register/minimum depth=3mm]
            qubit {} q[3];

            swap (q[0,1]); 
            swap (q[0,2]);
            not q[0] | q[1-2];
            not q[1-2] | q[0];
            not q[0] | q[1-2];
            h q[0];
        \end{yquant}
    \end{tikzpicture}
}

\newcommand{\dfive}{
    \begin{tikzpicture}[semithick]
        \begin{yquant}[register/minimum height=3mm, register/minimum depth=3mm]
            qubit {} q[3];

            swap (q[0,1]); 
            swap (q[0,2]);
            not q[1] | q[0]; 
            h q[0];
        \end{yquant}
    \end{tikzpicture}
}

\newcommand{\dsix}{
    \begin{tikzpicture}[semithick]
        \begin{yquant}[register/minimum height=3mm, register/minimum depth=3mm]
            qubit {} q[3];

            swap (q[0,1]); 
            swap (q[0,2]);
            not q[2] | q[0]; 
            h q[0];
        \end{yquant}
    \end{tikzpicture}
}

\newcommand{\dseven}{
    \begin{tikzpicture}[semithick]
        \begin{yquant}[register/minimum height=3mm, register/minimum depth=3mm]
            qubit {} q[3];

            swap (q[0,1]); 
            not q[1] | q[0];
            swap (q[0,2]);
            not q[2] | q[0];
            not q[0] | q[1-2];
            not q[2] | q[0-1];
            h q[0];
        \end{yquant}
    \end{tikzpicture}
}

\newcommand{\prfcircuit}{
    \begin{tikzpicture}[semithick]
        \begin{yquant}[register/minimum height=4mm, register/minimum depth=4mm, operator/separation=4mm]
            qubit {} q[3];
            [name=kzero]  z q[0];
            [name=kone]   z q[1];
            [name=ktwo]   z q[2];
            [name=kthree] zz (q[-1]);
            [name=kfour]  zz (q[0,2]);
            [name=kfive]  zz (q[1-]);
            [name=ksix]   zz (q[-]);
        \end{yquant}
       
        \node[fit=(kzero-0), draw, dotted, rounded corners, thick, inner ysep=3pt, inner xsep=3pt]{};
        \node[fit=(kone-0), draw, dotted, rounded corners, thick, inner ysep=3pt, inner xsep=3pt]{};
        \node[fit=(ktwo-0), draw, dotted, rounded corners, thick, inner ysep=3pt, inner xsep=3pt]{};
        \node[anchor=north, xshift=-0.3cm, yshift=0.5cm] at (kzero.north west){$\ell_0$};
        \node[anchor=north, xshift=-0.3cm, yshift=0.5cm] at (kone.north west){$\ell_1$};
        \node[anchor=north, xshift=-0.3cm, yshift=0.5cm] at (ktwo.north west){$\ell_2$};
        
        \node[fit=(kthree-0), draw, dotted, rounded corners, thick, inner ysep=19pt, inner xsep=6pt, yshift=-0.45cm,  "$\ell_3$"]{};
        \node[fit=(kfour-0), draw, dotted, rounded corners, thick, inner ysep=32pt, inner xsep=6pt, yshift=-0.9cm,  "$\ell_4$"]{};
        \node[fit=(kfive-0), draw, dotted, rounded corners, thick, inner ysep=19pt, inner xsep=6pt, yshift=-0.45cm,  "$\ell_5$"]{};
        \node[fit=(ksix-0), draw, dotted, rounded corners, thick, inner ysep=32pt, inner xsep=6pt, yshift=-0.9cm,  "$\ell_6$"]{};
    \end{tikzpicture}
}

\newcommand{\indcpanoisegame}{
    \begin{tikzpicture}
        \draw (0,0) node[rectangle, minimum height=0.7cm, minimum width=2cm, draw] (pkgen) {$\texttt{PKGen}(k)$};
        \draw (0, -2) node[rectangle, minimum height=2cm, minimum width=2cm, draw] (adversary) {$\mathcal{A}$};

        \draw[<->] (adversary.north) -- (pkgen.south);
        
        \draw[->] (adversary.south) -- (0, -3.65);
        \draw (0, -3.95) node(guess) {$\text{guess}$};
        
        \draw (-4, -1.5) node(rpk) {$(r^*, \rho k^{*})$};
        \draw[->] (rpk.east) -- ([yshift=0.5cm]adversary.west);
        
        \draw[->] (adversary.west) -- node[midway, above]{$m_0, m_1$}([yshift=-0.5cm]rpk.east);

        \draw node[below=of rpk.south east, anchor=north east, yshift=0.65cm](pkencrpk) {$\texttt{PKEnc}((r^*, \widetilde{\rho k^*}), m_b)$};
        \draw[->] ([yshift=-1cm]rpk.east) -- node[midway, above]{$c$}([yshift=-0.5cm]adversary.west);
    \end{tikzpicture}
}

\newcommand{\pubkeyscheme}{
    \begin{tikzpicture}
        \draw (0, 0) node[rectangle, minimum height=0.7cm, minimum width=2cm, draw] (skgen) {$\texttt{SKGen}$};
        \draw (0, -1.5) node[rectangle, minimum height=0.7cm, minimum width=2cm, draw] (pkgen) {$\texttt{PKGen}$};
        \draw (0, -4.5) node[rectangle, minimum height=0.7cm, minimum width=2cm, draw] (pkenc) {$\texttt{PKEnc}$};

        \draw[->] (skgen.south) -- node[midway, right]{$k$} (pkgen.north);
        
        \draw (0, -3) node (pubkeys) {$(r^{(1)}, \ket{\psi_{k, r^{(1)}}}), (r^{(2)}, \ket{\psi_{k, r^{(2)}}}), \cdots, (r^{(c)}, \ket{\psi_{k, r^{(c)}}})$};

        \foreach \i in {2, 0, -2} {
            \draw[->] (pkgen.south) -- ([xshift=\i cm]pubkeys.north);
        }

        \draw[->] (pubkeys.south) -- (pkenc.north);

        \draw (-3, -4.5) node (message) {$m$};
        \draw (3, -4.5) node (ciphertext) {$c$};
        \draw [->] (message.east) -- (pkenc.west);
        \draw [->] (pkenc.east) -- (ciphertext.west); 
    \end{tikzpicture}
}

\usepackage{dsfont}
\usepackage[compat=0.6]{yquant}
\useyquantlanguage{groups}

\usepackage{cryptocode}
\createprocedureblock{procb}{boxed}{}{}{}

\newcommand{\iden}{\ensuremath \mathds{1}}
\newcommand{\prf}[1]{\ensuremath {F_k\left(r, g, #1\right)}}

\usepackage{amsthm,amsmath,amsfonts,braket,algorithm,graphicx, braket,balance}

\theoremstyle{definition} \newtheorem{define} {Definition} [section]
\newtheorem {theorem} {Theorem}

\usepackage{ifthen}
\newboolean{fullversionflag}
\setboolean{fullversionflag}{true}
\newcommand{\fullversion}[2]{\ifthenelse{\boolean{fullversionflag}}{{#1}}{{#2}}}

\newcommand{\kb}[1]{\left[#1\right]}

\newcommand{\trd}[1]{\left|\left| #1 \right| \right|}

\newcommand{\Hmin}{H_\infty}

\newcommand{\up}[1]{^{({#1})}}

\floatname{algorithm}{Protocol}

\begin{document}

\title{Quantum Public Key Encryption for NISQ Devices}

\author{
    \IEEEauthorblockN{Nishant Rodrigues}
    \IEEEauthorblockA{
        \textit{Microsoft Quantum} \\
        Redmond, WA USA \\
        nirodrigues@microsoft.com \orcidlink{0000-0001-8882-4677}
    }
    \and
    \IEEEauthorblockN{Walter O. Krawec}
    \IEEEauthorblockA{
        \textit{University of Connecticut} \\
        Storrs, CT USA \\
        walter.krawec@uconn.edu
    }    
    \and
    \IEEEauthorblockN{Brad Lackey}
    \IEEEauthorblockA{\textit{Microsoft Quantum} \\
    Redmond, WA USA \\
    brad.lackey@microsoft.com\orcidlink{0000-0002-3823-8757}
    }
    \and
    \IEEEauthorblockN{Deb Mukhopadhyay}
    \IEEEauthorblockA{\textit{Microsoft Security} \\
    Redmond, WA USA \\
    debasish.mukhopadhyay@microsoft.com
    }    
    \and
    \IEEEauthorblockN{Bing Wang}
    \IEEEauthorblockA{
        \textit{University of Connecticut} \\
        Storrs, CT USA \\
        bing@uconn.edu
    }    
}

\maketitle

\begin{abstract}
Quantum public-key encryption (PKE), where public-keys and/or ciphertexts can be quantum states, is an important primitive in quantum cryptography. Unlike classical PKE (e.g., RSA or ECC), quantum PKE can leverage quantum-secure cryptographic assumptions or the principles of quantum mechanics for security. It has great potential in providing for secure cryptographic systems under potentially weaker assumptions than is possible classically. In addition, it is of both practical and theoretical interest, opening the door to novel cryptographic systems not possible with classical
information alone. While multiple quantum PKE schemes have been proposed, they require a large number of qubits acting coherently, and are not practical on current noisy quantum devices. In this paper, we design a practical quantum PKE
scheme, taking into account the constraints of current NISQ devices. Specifically, we design a PKE scheme with quantum-classical public keys and classical ciphertexts, that
is noise-resilient and only requires a small number of qubits acting coherently. In addition, our design provides tradeoffs in terms of efficiency
and the number of qubits that are required.
\end{abstract}

\section{Introduction}


Public Key Encryption (PKE) is ubiquitous in our society, allowing for the secure encryption of messages using a \emph{pubic key} (which multiple parties may hold copies of, including an adversary), in such a way that only the \emph{secret key} holder may decrypt them. However, currently its security relies on strong computational hardness assumptions.
In contrast,
quantum PKE, where public-keys and/or ciphertexts can be quantum states,  leverages quantum-secure cryptographic assumptions or the principles of quantum mechanics for security. 

Recent studies~\cite{Morimae2022,Barooti2023,Grilo2023,Barooti24:qPKE,coladangelo2023quantum} have shown drastically different results in  classical and quantum regimes regarding PKE. 
In the classical regime, based on Impagliazzo’s ``five worlds of complexity''~\cite{Impagliazzo1995},  PKE and one-way function (OWF) exist in two distinct worlds, specifically, Cryptomania and MiniCrypt, respectively. 
We are not sure if we live in the weaker  world of MiniCrypt, a place where OWFs exist but PKE systems are insecure, or the much stronger world of Cryptomania, where PKE systems exist. That is, if it turns out that we are in MiniCrypt, all our PKE systems (including widely used PKE schemes such as RSA and ECC, and even the post-quantum ones) are insecure.
In the quantum regime, on the contrary, 
these studies have shown that PKE schemes can be realized using only OWFs, \emph{or even weaker assumptions}.  Thus, quantum PKE may exist even if classical ones do not!
Clearly there is great potential in quantum PKE systems, as they will allow for public key systems to continue to operate even if our post-quantum assumptions fail. Beyond this, they also hold numerous, and fascinating, practical and theoretical interest, opening the door to novel cryptographic systems not possible with classical information alone. Furthermore, their construction, and proof of security, can hold broad benefit to the field of cryptography in general.

By now there are several interesting constructions for quantum PKE systems, ranging from work which had information theoretic security assuming the public key could be distributed securely \cite{nikolopoulos2008applications,vlachou2015quantum}, to newer works with security based on quantum one way functions or quantum pseudorandom function-like states (see \cite{Barooti2023,Grilo2023,ananth2022cryptography,Barooti24:qPKE} just to list a few).  Existing quantum PKE schemes can generally be classified into multiple categories, depending on whether the public key and ciphertext are quantum or classical. 
In the first category, the public key is quantum, while the ciphertext is classical~\cite{Barooti2023,Grilo2023,Barooti24:qPKE}. In the second category, both public key and ciphertext are classical, while the algorithm that is used for encryption is quantum~\cite{Okamoto2000}.  
In the third category, both public keys and ciphertext are quantum~\cite{Gottesman2005,Kawachi2005,coladangelo2023quantum}.
All existing quantum PKE schemes, to our knowledge, however, suffer from two major problems. First, the quantum public key must consist of a large number (typically thousands) of qubits operating together, coherently. Second, current day systems are highly intolerant of any noise in the public key and would require quantum error correction to operate successfully. 

In this paper, we focus on designing practical quantum PKE schemes using quantum-secure pseudorandom functions, taking into account the constraints on current NISQ (Noisy intermediate-scale quantum) devices.
Specifically, 
we focus on quantum PKE schemes
where the public key is classical-quantum, and the ciphertext is classical. This allows the ciphertext to be stored on classical devices, and the storage can be for a long time, unlike a quantum ciphertext, which can only be stored for a short amount of time due to decoherence of quantum memory.
In addition, considering that current quantum processors are noisy and only have a small number of qubits, we design a PKE scheme that is noise-resilient and only requires a small number of qubits acting coherently together (we require a large number of qubits, however they can be grouped into practical ``chunks'' and run sequentially or in parallel).  Furthermore, our scheme supports the use of classical error correction on the ciphertext, while maintaining security, allowing the protocol to operate even if the quantum public key is noisy.  To our knowledge this is the first PKE scheme using Pseudo Random Functions (PRFs) that can support this. In addition, our design provides  trade-offs in terms of efficiency and the number of qubits that used in the public key so as quantum computers become more reliable, one may create public keys with a larger number of qubits in superposition, resulting actually in an overall smaller public key. We design and implement our PKE schemes on current quantum computers, and evaluate the error rate, and the efficiency of the schemes. 


Specifically, our work makes the following contributions:
\begin{itemize}
    \item {\em Quantum PKE protocol.} Inspired by the Round Robin quantum key distribution protocol~\cite{sasaki2014practical} and a recent ``trap door'' quantum encryption protocol~\cite{coladangelo2023quantum}, we construct a novel quantum PKE protocol that is suitable for current NISQ devices. Specifically, instead of relying on keeping a large number of qubits in a coherent superposition state, our protocol only relies on multiple ``groups'' of a small number of qubits remaining in a superposition. In addition, these groups do not have to be created all at once, instead, they can be created sequentially and ``streamed'' to users. Last, our construction allows for errors in the quantum state and the use of basic classical error correction, instead of quantum error correction.  To our knowledge, this is the first quantum PKE system which utilizes PRFs that can be implemented, securely, using today's NISQ devices.
    \item {\em Security Proof.}    We present a rigorous security analysis of the protocol.  Our proof combines traditional cryptographic hybrid arguments with QKD style proof techniques (e.g., bounds on quantum min entropy).  This shows further evidence that QKD security analyses and techniques can have broad impact outside of basic key-distribution, towards the design and security analyses of novel quantum cryptographic protocols (such as PKE in our case).  We also define security for the case where public keys are noisy (which hasn't been considered in recent PKE literature).

    \item {\em Circuit design and Implementation.} We design and implement our proposed PKE protocol by designing suitable circuits to run the public key generation algorithm and the encryption algorithm (decryption is purely classical).  We test our circuits on the Quantinuum H1-1 computer, reporting on the error rates observed, along with the required public key size needed to encrypt a message under such scenarios.  We also release our source code for this experiment as open source \cite{sourcecode}.
    \end{itemize}

\section{Preliminaries}\label{sec:background}

Let $q \in \{1,\cdots, d\}^N$ be some word.  We write $q_i$ to mean the $i$th character of $q$.  We use $H(A)$ to mean the Shannon entropy of random variable $A$ while we write $H(x)$, for $x\in[0,1]$ to mean the binary Shannon entropy function, namely $H(x) = -x\log_2x - (1-x)\log_2(1-x)$.  Note, later, we will use a lower-case $h(\cdot)$ to mean a hash function.

Let $\rho_{AE}$ be a classical-quantum state, where the $A$ register consists of $N$-bits.  The \emph{quantum min entropy}, denoted $\Hmin(A|E)_\rho$ was defined in \cite{renner2008security} and is related to the maximal guessing probability of Eve (who owns the $E$ system) guessing the value of the $A$ system \cite{konig2009operational}.  The \emph{smooth quantum min entropy} \cite{renner2008security} is denoted $\Hmin^\epsilon(A|E)_\rho$, and is defined to be $\Hmin^\epsilon(A|E)_\rho = \sup_\sigma \Hmin(A|E)_\sigma$, where the supremum is over all density operators $\sigma$ such that $\trd{\rho-\sigma}\le \epsilon$.

Consider the privacy amplification process: namely, a random two-universal hash function $h:\{0,1\}^N\rightarrow\{0,1\}^\ell$ is chosen and the $A$ register is sent through the given hash, yielding a new cq-state $\sigma_{KE}$ where the $K$ register is $\ell$-bits long (the result of mapping the $A$ register through $H$) and the new $E$ register contains a description of the hash function chosen.  Then, it was shown in \cite{renner2008security}, that:
\begin{equation}\label{eq:intro:PA}
\trd{\sigma_{KE} - 2^{-\ell}I_K\otimes\sigma_E} \le 2^{-\frac{1}{2}(\Hmin^\epsilon(A|E)_\rho - \ell)} + 2\epsilon,
\end{equation}
Essentially, the above ensures that, so long as the min entropy is ``high'', the final output after mapping the $A$ register through a two-universal hash function, will be close, in trace distance, to an ideal key: one that is chosen uniformly at random and is independent of any other system (in particular the $E$ system controlled by the adversary).  Such an ideal key could be used as a one time pad (OTP) encryption key for instance.

\subsection{Quantum Secure Pseudo Random Functions (PRFs)}\label{sec:PRF}
Our protocol will make use of quantum secure PRFs.
A PRF is a keyed function $F_k$ mapping $N$ bit strings to $n$ bit strings.  We assume, throughout, that the key is of size $\lambda$-bits where $\lambda$ will be our security parameter.  In general, a PRF also consists of a secret key generation function which produces a valid secret key $k$ for the PRF family.  To define the security of a PRF, we need to consider the attack model. Classically, we consider an adversary that is allowed to make classical queries to $F_k$ and the responses are indistinguishable (from a computational point of view) from the responses produced by simply choosing an arbitrary function $f:\{0,1\}^N\rightarrow \{0,1\}^n$.  See \cite{goldreich1986construct}.

In the quantum setting, however, the adversary is allowed to make superposition queries.  In particular, consider a quantum polynomial time algorithm $\adv$.  We write $\adv^{F_k}$ to mean an algorithm with oracle access to $F_k$ in superposition.  Namely, $\adv$ is allowed to evaluate the unitary operator $U_{F_k}\ket{x,y} = \ket{x, F_k(x)\oplus y}$.  Of course, $\adv$ is not given the circuit for $U_{F_k}$; instead the adversary is allowed to execute the unitary in a black-box fashion.  The adversary may, of course, also evaluate the unitary on a superposition (called a \emph{superposition query}).  Compare that with $\adv^f$ which is the same adversary algorithm, but now which is able to evaluate the unitary $U_f\ket{x,y} = \ket{x, f(x)\oplus y}$.  The idea behind a quantum secure PRF is that the adversary will be unable to distinguish the case when it is interacting with $U_{F_k}$ as opposed to $U_f$, for randomly chosen $k$ or $f$. See \cite{zhandry2021construct} for more details.

Formally, a PRF is said to be \emph{quantum secure} if the following holds:
\begin{define}\label{def:intro:PRF}
  (From \cite{zhandry2021construct}, see also \cite{ananth2022pseudorandom}): Let $F_k:\{0,1\}^N\rightarrow \{0,1\}^n$ be a keyed function with key-size $\lambda$.  Then, $F_k$ is a quantum secure PRF if for every quantum polynomial time algorithm $\adv$ which outputs a single classical bit, there exists a negligible function $\nu$ such that:
  \begin{equation}
    |\Pr(\adv^{F_k}=1) - \Pr(\adv^{f}=1)| \le \nu(\lambda).
  \end{equation}
  Above, the first probability is over all choices of key $k$ and any randomness inside $\adv$ while the second probability is over the choice of all possible functions $f:\{0,1\}^N\rightarrow\{0,1\}^n$ along with any randomness in $\adv$.  Note that the adversary is allowed superposition queries to the oracle function as discussed above.
\end{define}

\subsection{Round-Robin Quantum Key Distribution (RR-QKD) Protocol}\label{sec:RRQKD}
Our protocol will also make use, as a backbone, a particular QKD protocol known as the \emph{Round-Robin protocol} first introduced in~\cite{sasaki2014practical}.  This protocol is interesting as it allows one to bound an adversary's information gain based solely on the dimension of the system, $d$, as opposed to bounding this based on the amount of noise in the channel.  Whenever $d \ge 3$, the system can produce a secure secret key (when $d=2$, the protocol is actually insecure).  We discuss the protocol here only briefly, followed by an important security result we will use later.

The protocol works with $d$-dimensional qudits over multiple rounds.  For each round $1, 2, \cdots, N$, Alice will choose a random bit string $s_1, \cdots, s_d$ and create the quantum state:
\begin{equation}\label{eq:intro:init-state-rr}
  \ket{\psi_s} = \frac{1}{\sqrt{d}}\sum_{t=1}^{d}(-1)^{s_t}\ket{t}.
\end{equation}
Note that for each round, Alice chooses a new $s_1$ through $s_d$ uniformly and independently at random.  Normally, for a QKD system, such a state can be created by a single photon using time-bin encoding \cite{sasaki2014practical}. However, for our work, we will produce such states on a NISQ device where the summation will be over bit strings and we will have $d=2^n$ for some, relatively small, $n$.  Such a state can be readily created using basic quantum gates.  We will discuss our implementation details, later, in Section~\ref{sec:design}.

Returning to the QKD protocol, the above state is sent to Bob who, on receipt of the signal, will choose a $\delta \in \{1, \cdots, d-1\}$ and perform a measurement which returns two data points: first, an index $i\in\{0, \cdots, d-1\}$, and second a bit $w\in\{0,1\}$.  In the noise-less case, it should hold that $i$ is uniform random, and $w = s_i\oplus s_j$ for some $j$ such that $j-i \equiv \delta \text{ mod } d$.  The value of $j$ can be determined from the value of $i$ and $\delta$.  For specific details on how this measurement is implemented in the QKD protocol, the reader is referred to~\cite{sasaki2014practical}.  We will discuss our implementation of this measurement later in Section~\ref{sec:design}.

Bob will send to Alice the value of $i$ and $j$ for each round, keeping the value of $w$ secret.  This will allow Alice and Bob to share a correlated raw key bit for that round.  This repeats until a sufficient number of raw key bits have been distilled.  We will use $W$ to denote the random variable storing Bob's value of $w$ and $W'$ to be Alice's value.

After this, an error correction protocol is run in order to correct for errors in the string $w$ held by Alice and Bob.  For this, one needs an upper bound on the number of bit-flip errors expected in these two random strings.  Note, however, they do not need to actively observe the error in alternative bases to bound Eve's information, as is needed in other QKD protocols.  Following error correction, privacy amplification is run, which involves Alice (or Bob) choosing a random two-universal hash function $h:\{0,1\}^N\rightarrow \{0,1\}^\ell$ (where $N$ is the bit-size of $w$) and computing $h(w)$.  This will be their secret key.

The following theorem will be useful later in our proof, which bounds the quantum min entropy given a single round of the system:
\begin{theorem}\label{thm:intro:RR-rate}
  (From \cite{skoric2017short}): Assume an adversary is given the state $\ket{\psi_s}$ (shown in Equation (\ref{eq:intro:init-state-rr})) and Bob is given any arbitrary state produced by the adversary.  Alice and Bob run the RR-QKD protocol for this one round, resulting in random variables $I$, $J$, and $W'$ (the latter of which is Alice's secret value $s_i\oplus s_j$).  Eve has a quantum ancilla $E$ and is given copies of $I$ and $J$.  Then, it holds that:
  \begin{equation}
    \Hmin(W'|IJE) = 1 - \log_2\left(1+ \frac{2}{d}\right).
  \end{equation}
\end{theorem}
\begin{proof}
  The proof can be found in \cite{skoric2017short} and is based on properties of min entropy, and finding the eigenvalues of the resulting density operator for one round of the protocol.
\end{proof}
Note that, when $d=2$ (i.e., each round uses a single qubit), the protocol is insecure.  This is not difficult to see as in this case, Alice is essentially encoding her key value into orthogonal states.  However, as soon as $d > 2$, security is possible.  The interesting thing about this protocol is that Alice and Bob don't have to measure the signal noise to know how much information Eve has.  They only need to know the final bit-flip error rate (in their values $W$ and $W'$) for error correction -- this knowledge is not needed to bound Eve's uncertainty, which is in stark contrast to all other QKD protocols which generally require monitoring the noise in multiple bases to bound $\Hmin(A|E)$.

\section{Quantum PKE Protocol}\label{sec:protocol}
Our protocol takes inspiration from the Round-Robin QKD protocol~\cite{sasaki2014practical} discussed in Section~\ref{sec:background}, and a recent ``trap door'' quantum encryption protocol from~\cite{coladangelo2023quantum}.  It is illustrated at a high level in Figure \ref{fig:protocol:high-level} and in more detail in \figurename~\ref{fig:protocol_enc_dec}. As we shall see, our protocol can encrypt messages of size $\ell$, where $\ell$ can be arbitrarily large, but fixed.  Let $F$ be a quantum secure PRF according to Definition~\ref{def:intro:PRF} with security parameter $\lambda$, which maps inputs of size $p + \lceil\log_2 N\rceil + n$ to single bit outputs, where $N$ is to be determined and will be a function of $\ell$, and $n$ will be a parameter that the users may optimize over.  The value of $p$ may be a function of $\lambda$ and will relate to the probability of failure in the event two public keys have the same classical tag $r$.  We may write $F_k(r, g, x)$, where $r \in\{0,1\}^p$, $g \in \{0,1\}^{\lceil\log_2N\rceil}$, and $x \in \{0,1\}^n$.  The protocol is defined as follows:
\begin{itemize}%
  \item \texttt{SKGen}: will choose a random secret key $k$ using the key generation function of the given PRF.
  \item \texttt{PKGen}: will choose a random $p$-bit string $r$ and output the public key $(r, \ket{\psi_{k,r}})$, where:
    \begin{equation}\label{eq:protocol:public_key}
      \ket{\psi_{k,r}} = \bigotimes_{g=1}^N \frac{1}{\sqrt{2^n}}\left(\sum_{x\in\{0,1\}^n}(-1)^{F_k(r,g,x)}\ket{x}\right).
    \end{equation}
    The above state can be created efficiently.
    For the above, we refer to a \emph{group} as one of the particular states; namely group $g$ is $\frac{1}{\sqrt{2^n}}\left(\sum_{x\in\{0,1\}^n}(-1)^{F_k(r,g,x)}\ket{x}\right)$.  Note that each individual group is separable from the others.  In general, we will see that $n$ can be small (three to five qubits), while $N$ must be large.  This makes our system practical on today's NISQ devices.
  \item \texttt{Enc}: will take in a public key $(r, \ket{\psi_{k,r}})$ and a message $m\in\{0,1\}^\ell$ and perform the following steps:
    \begin{enumerate}%
    \item A random string $\delta \in \{1, 2, \cdots, 2^{n}-1\}^N$ is chosen, with $\Pr(\delta_g = y) = 1/(2^{n}-1)$, where $\delta_g$ is the value for the $g$th group.
    \item For each group $g = 1, \cdots, N$, a measurement is performed using the RR-QKD POVM \cite{sasaki2014practical} $\{M^{\delta_g}_{i,w}\}$ where:
      \begin{equation}\label{eq:protocol:enc_measurement}
        M^{\delta_g}_{i,w} = \frac{1}{2}P\left(\ket{i} + (-1)^w\ket{j}\right),
      \end{equation}
      where $P(\ket{z}) = \ket{z}\bra{z}$, and where, above, $j - i = \delta_g \pmod{2^n}$.  Note that we assume $j$ can be determined from the measurement outcome;  we shall show a possible implementation of this measurement in Section~\ref{sec:design}.  Note that an observation of $(i,w)$ will occur only if the relative phase difference between $\ket{i}$ and $\ket{j}$ in the $g$-th group is equal to $(-1)^w$.  Ideally, it should hold that $F_k(r,g,i) \oplus F_k(r,g,j) = w$. However, errors in the public key may cause $w$ to be random on some groups.  We call $w$ the \emph{measured pad value}.
    \item Let $i,j,w$ now be the resulting measurement strings after performing all measurements on all groups using the above POVM, that is, $i,j\in \{0, \cdots, 2^n-1\}^N$ and $w \in \{0,1\}^N$.  A random two universal hash function $h$ is chosen such that $h:\{0,1\}^N \rightarrow \{0,1\}^\ell$. Finally, the ciphertext, encrypting message $m$, is output:
      \begin{equation}
        c = (r, i, j, h, h(w)\oplus m)
      \end{equation}
    \end{enumerate}
  \item \texttt{Dec}: will take the secret key and a cipher text $c$ and compute $w' = w'_1\cdots w'_N$ where $w'_g = F_k(r, g, i_g) \oplus F_k(r,g,j_g)$.  The message, then, can be recovered  by passing $w'$ through the hash function $h$ (provided in the ciphertext) and XOR'ing with the last component of the ciphertext.  We call $w'$ the \emph{computed pad value}.
  \end{itemize}

  \begin{figure}
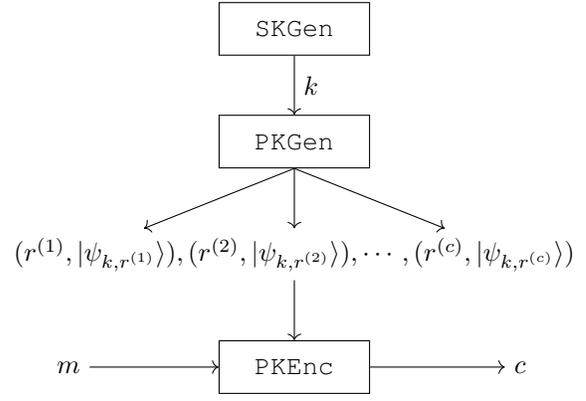

    \centering
    \pubkeyscheme 
    \caption{A high level diagram showing how the PKE scheme is used.  First a secret key is produced using \texttt{SKGen}.  From this multiple copies of the quantum public key are created using \texttt{PKGen}.  Any one of these can be used to encrypt a message.}\label{fig:protocol:high-level}
  \end{figure}

  Note that if there are errors in the public key, it will not be the case that $w = w'$, where $w$ is the measurement result during encryption and $w'$ is the value computed at decryption based on the PRF evaluation.  However, it should hold that if the public key error is ``low enough'', $w$ and $w'$ should be highly correlated.  
  In fact, our construction supports the additional step of error correction on the string $w$.  Namely, if it is known that the number of errors between $w$ and $w'$ are no larger than $Q$, the encryptor may add a forward error correcting code (e.g., an LDPC code~\cite{Elkouss2009:LDPC,MinkNakassis2012:LDPC}), allowing the decryptor to recover $w$ given $w'$ and the codeword.  This code can be placed in the ciphertext; note that doing so will ``leak'' information to Eve, however the two-universal hash function can ensure secrecy.  Note also that error correction can be done interactively during decryption (e.g., using Cascade protocol~\cite{BrassardSalvail1994:Cascade}). However that would require the encryptor to also be active at the time of decryption, which may not be suitable in applications.

  It is clear that our proposed protocol is based on the Round-Robin (RR) QKD protocol.  To relate the two, the ``Groups'' are individual rounds of the RR protocol, while the value of $2^n$ is the dimension of the quantum state used for each round ($d = 2^n$).  Later, in our circuit design and implementation (Section~\ref{sec:design}), we will use $n = 3$ which provides a balance of efficiency and practicality, however we will prove security for any $n > 1$ in Section~\ref{sec:security}.  Finally, $p$ will be related to the overall probability of failure based on how many public keys are created.  We will set $p$ to be large in order to ensure a negligible failure probability (e.g., $p=256$).  Finally, $N$, the number of groups, can be set based on the message size, the dimension of each group $n$, the expected noise, and the security parameter as we prove in our Theorem~\ref{thm:security} below.

  Overall, our protocol has several advantages over prior quantum public key states.  Our protocol does not rely on keeping a large number of qubits in a coherent superposition state.  Instead, it relies on multiple ``groups'' of a small number of qubits remaining in a superposition.  Furthermore, these groups do not have to be created at once if the parties wish to encrypt larger messages (the public key can be ``streamed'' to users).  Finally, our construction allows for errors in the quantum state and the use of basic classical error correction, as opposed to the need for quantum error correction, necessary in all other quantum public key encryption systems that we are aware of.  This makes our construction the first quantum public key system, to our knowledge, that can be implemented in a secure manner using today's quantum computers.  In Section~\ref{sec:eval}, we demonstrate this by testing our system on Quantinuum's quantum computers.   Note that our work is not simply a trivial application of PRFs to the RR-QKD protocol.  We must derive a suitable security definition for noisy public keys, and prove security in the event multiple copies of a public key are available to the adversary (which is not a requirement in the RR-QKD protocol). Finally, we must also implement circuits to perform the needed measurements on a NISQ device.
  

\begin{figure*}[!t]
    \centering
    \procb[colspace=-2cm]{}{
            \textbf{Alice} \> \< \textbf{Bob} \\
            k \leftarrow \texttt{SKGen}(\lambda) \> \< \\
            r, \ket{\psi_{k,r}} = \bigotimes_{g=1}^N \frac{1}{\sqrt{2^n}}\left(\sum_{x\in\{0,1\}^n}(-1)^{\prf{x}}\ket{x}\right)\leftarrow \texttt{PKGen}(k) \>\< \\
            \> \sendmessageright*[6cm]{(r, \ket{\psi_{k,r}})} \< \\
            \> \< m \in \{0, 1\}^l \\
            \> \< \delta \leftarrow \{1, \cdots, 2^{n}-1\}^N\\
            \> \< i, j, w \leftarrow \mathsf{Measure}(\ket{\psi_{k, r}}) \text{ according to } (\ref{eq:protocol:enc_measurement}) \\
            \> \< \text{Sample } h \text{ from a two-universal hash function family} \\
            \> \< c = (r, i, j, h, h(w) \oplus m) \\
            \> \sendmessageleft*[6cm]{c} \< \\
            r, i, j, h, c' \leftarrow \mathsf{Unpack}(c) \> \<  \\
            w' = \left(w'_g = \prf{i_g} \oplus \prf{j_g} \right)_{g=1}^{N} \> \< \\
            m' = h(w') \oplus c' \> \<
        }
    \caption{Illustration of the proposed quantum public key encryption protocol.}
    \label{fig:protocol_enc_dec}
\end{figure*}


\section{Security Analysis}\label{sec:security}
We consider the \texttt{IND-CPA-EO} model introduced in~\cite{Barooti24:qPKE} for quantum public keys.  In particular, this definition involves an oracle who creates a challenge public key and keeps this private.  An adversary, later, sends two challenge messages to the oracle.  One of these messages is encrypted using the challenge public key and the resulting ciphertext is sent to the adversary.  The adversary must guess which message was encrypted.  In this definition, the ``\texttt{EO}'' extension signifies access to an encryption oracle with access to the challenge public key.  However, this is only relevant for public key systems which allow a quantum public key to be used for multiple encryptions, which ours does not (i.e., while there may be multiple copies of the public key, each individual copy can only be used to encrypt a single message).  Thus, we do not consider the $\texttt{EO}$ version.  See \cite{Barooti24:qPKE} for more details.

Now, in our system, the challenge public key consists of a classical part; furthermore the public key may suffer from noise (e.g., if it is created and/or stored on a NISQ device or teleported across a noisy network).  To model this, we will make two extensions to the \texttt{IND-CPA} definition from \cite{Barooti24:qPKE}; our extensions will be to the adversary's advantage (and, thus, anything proven secure using our definition will be secure in the original).  Our extensions are as follows:
\begin{enumerate}%
 \item First, when creating the challenge public key (which will be used to encrypt the challenge ciphertext), the oracle makes two identical copies, sending one to the adversary (including the classical and quantum part) while keeping the second copy private.  Note that, typically, the challenge public key is kept entirely private, so this alteration can only be to the adversary's advantage.  We make this alteration as the classical portion ensures each new public key is highly unique and so by giving the adversary a copy of the actual challenge public key, we are ensuring that security holds even if Eve manages to secure a copy of the challenge public key somehow.
 \item Second, to model noise in the public key storage, we send the challenge public key through a CPTP map $\mathcal{E}$ producing a noisy version of the public key.  When the challenge message is encrypted, it will be encrypted using this noisy version (thus, error correction information must be provided in the ciphertext which could be beneficial to the adversary).  We will assume this noise map is independent of Eve and models natural noise only.  Interesting future work may be found in the case where the adversary has some partial influence over this map (perhaps entangling with the output also).
 \end{enumerate}


With this in mind, we define the following security definitions.  As mentioned, these definitions take the \texttt{IND-CPA-EO} model introduced in~\cite{Barooti24:qPKE} and modify them by adding additional capabilities to the adversary (thus they are stronger definitions in the sense that any system which are proven secure using our definition will automatically be secure according to the definitions given in~\cite{Barooti24:qPKE}).  For the case where the noise map is independent of Eve and given by a CPTP map $\mathcal{E}$, we define the following notion of \cpanet security (note the need to specify the CPTP map) which involves the following game played between an oracle and an adversary (see also Figure \ref{fig:security:cpa-security}):

\begin{enumerate}%
\item The oracle runs $\texttt{SKGen}(\lambda)$ to produce a secret key $k$.  The oracle then runs $\texttt{PKGen}(k)$ to produce the challenge public key $(r^*,\rho k^*)$.  Given $r^*$ and $k$, the oracle produces a second copy of this challenge public key (both the classical and quantum part); this second copy is given to the adversary $\adv$.  The quantum part of the first challenge public key is sent through the given CPTP map $\mathcal{E}$ to produce a noisy version $\widetilde{\rho k}^* = \mathcal{E}(\rho k^*)$.
\item The adversary $\adv$ is given oracle access to $\texttt{PKGen}(k)$.
\item The adversary outputs two classical messages $m_0$ and $m_1$ of equal length.
\item The oracle chooses a random bit $b$ and encrypts $m_b$ using $(r^*, \widetilde{\rho k}^*)$.  Namely the oracle computes $c = \texttt{PKEnc}( (r^*,\widetilde{\rho k}^*), m_b)$.  This ciphertext (called the \emph{challenge ciphertext} is send to the adversary).
\item The adversary has continued oracle access to $\texttt{PKGen}(k)$.
\item The adversary outputs a bit $b'$.  We say the adversary \emph{wins} if $b' = b$ (in this case, the experiment outputs a $1$).
\end{enumerate}

If we denote the above experiment, for a given public key protocol $\prod$ and a specific adversary $\adv$, by $\experiment_{\prod,\adv}$.  We say a given protocol is \cpanet secure if the probability of winning is negligibly close to $1/2$.  Namely:
\begin{define}
  Let $\prod$ be a quantum public key encryption protocol with security parameter $\lambda$, and $\mathcal{E}$ be some CPTP map.  Then we say $\prod$ is \cpanet if for every quantum polynomial time adversary $\adv$, there exists a negligible function $\nu$ such that:
    $\Pr\left(\experiment_{\prod,\adv} = 1\right) \le \frac{1}{2} + \nu(\lambda).$
\end{define}

\begin{figure}
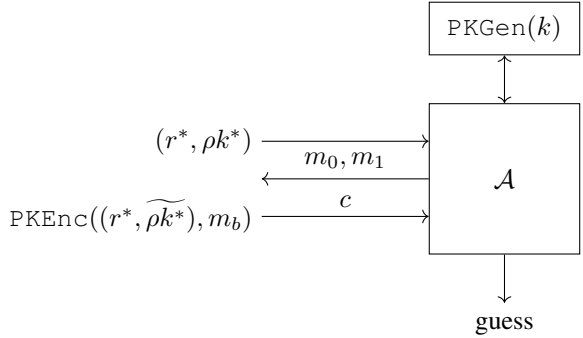

  \centering
  \indcpanoisegame 
  \caption{Showing the $\cpanet$ game.  The adversary has oracle access to $\texttt{PKGen}$ to produce as many copies of the public key as desired.  A challenge public key is also created, one copy of which is given to the adversary.  Later, the adversary outputs two classical messages and one of them is encrypted, randomly, using a noisy version of the challenge public key.  Finally, the adversary must guess which ciphertext was created.}\label{fig:security:cpa-security}
\end{figure}

We now prove that our protocol described in Section~\ref{sec:protocol} is \cpanet secure for any CPTP map $\mathcal{E}$ that satisfies the following property: given a public key input of the form in Equation \ref{eq:protocol:public_key}, if one were to perform measurements as described by the encryption POVMs (Equation \ref{eq:protocol:enc_measurement}), then the probability that the given output does not equal the expected output (namely the expected sum, modulo two, of the PRF function), is at most $Q$ for some known $Q < 1/2$ and each group is treated independently and identically by $\mathcal{E}$.

\begin{theorem}\label{thm:security}
Let $\ell \ge 1$ be the message size in bits. Then, if we set $N$ to be:
  \begin{equation}
      N = \left\lceil\max\left(0,\frac{\ell - 2\log_2\frac{5c^2}{2^p}}{1 - \log_2(1+2^{1-n}) - H(Q)}\right)\right\rceil,
  \end{equation}
   where $H(Q)$ is the binary Shannon entropy of $Q$ and where $c$ is the maximal number of public keys created (or, equivalently, requested by $\adv$),
  then our protocol, above, is $\cpanet$ secure for any CPTP map $\mathcal{E}$ described above.  In particular, for every quantum polynomial time adversary $\adv$, there exists a negligible function $\nu(\lambda)$, such that:
  \begin{equation}
  \Pr(\experiment_{\prod,\adv} = 1) \le \frac{1}{2} + \frac{c^2}{2^p} + \nu(\lambda).
  \end{equation}
  In particular, if $p = \eta\lambda$, for some constant $\eta$, and $c = poly(\lambda)$, the probability of winning is negligible close (in $\lambda$) to $1/2$.
\end{theorem}
\begin{proof}
  We use a standard hybrid argument to prove security.  Fix an adversary $\adv$.  First, we will show that the PRF used in our construction may be replaced with a random function.  Let $\prod^{F_k}$ be the original protocol (defined in Section~\ref{sec:protocol}) and consider the following protocol $\prod^f$ which is identical to $\prod^{F_k}$ except that:
  \begin{enumerate}%
  \item $\texttt{SKGen}(\lambda)$ will choose a random function $f:\{0,1\}^p\times\{0,1\}^{\log N}\times\{0,1\}^{n} \rightarrow \{0,1\}$ uniformly at random from all possible functions.
  \item $\texttt{PKGen}(\lambda)$ will create a public key state by evaluating $f(r,g,x)$ instead of $F_k(r,g,x)$.
  \end{enumerate}

  Since $F_k$ is a quantum secure PRF according to Definition \ref{def:intro:PRF}, it is not difficult to show that the probability of winning $\prod^f$ is negligibly close to winning $\prod^{F_k}$.  Formally, consider the following distinguisher $D$ which will simulate the $\cpanet$ oracle.  First, $D$ will construct two challenge public keys by choosing $r^*$ uniformly at random from $\{0,1\}^\lambda$ and then constructing the quantum state:
  \begin{equation}
    \ket{\psi_g} = \ket{r^*,g}\sum_{x\in\{0,1\}^n}\ket{x}\ket{-},
  \end{equation}
  for $g = 1, 2, \cdots, N$.  The distinguisher will send each $\ket{\psi_g}$ to its oracle which will execute the unitary operation $U_f$ or $U_{F_k}$ as discussed in Section \ref{sec:PRF}.  In the case the oracle is performing $U_f$, the state will evolve to $\ket{\psi_g'} = \ket{r^*,g}\sum_x(-1)^{f(r^*,g,x)}\ket{x,-}$. Otherwise it will be of the form $\ket{\psi_g'} = \ket{r^*,g}\sum_x(-1)^{F_k(r^*,g,x)}\ket{x,-}$.   The distinguisher will discard the last qubit register of each state along with the initial qubits storing the ``$r^*$'' and ''$g$'' values thus creating the final quantum public key.  This is repeated twice, with the second copy being given to the adversary $\adv$.  The first copy's quantum portion is run through $\mathcal{E}$.

  The distinguisher $D$ will then run the adversary $\adv$; whenever it requests a public key the above process is run, choosing independent random $r$ values each time.  When the adversary outputs a message $m_0$ and $m_1$, the distinguisher will pick one of them (denoted $b$) to encrypt using the previously created challenge public key.  The resulting ciphertext is sent to the adversary.  When the adversary finally outputs a bit $b'$ (a guess as to which message was encrypted), the distinguisher will output $1$ only if $b' = b$ (i.e., if the adversary wins).  Since the PRF is quantum secure according to Definition \ref{def:intro:PRF}, there exists a negligible function $\nu_1(\lambda)$ such that the distinguishing advantage is less than $\nu_1(\lambda)$; hence the probability of $\adv$ winning in $\prod^{F_k}$ can be only $\nu_1(\lambda)$ different from winning in $\prod^{f}$.

  It remains to be shown that $\Pr(\experiment_{\prod^{f}, \adv}) \le \frac{1}{2} + \nu_2(\lambda)$ for some negligible function $\nu_2$.  For this, we first note that it suffices to consider an adversary that requests $c=poly(\lambda)$ public keys immediately on execution and then never calls the public key generation oracle again.  This is clear, since the public key generation oracle will run independently of any cipher text creation and simply creates public keys independently of the execution of the experiment.  Thus, if there were an advantage to requesting additional public keys in the middle of the experiment (e.g., after the adversary outputs two challenge messages), the adversary may simply request multiple public keys at the start, and save them in memory to use later as needed.  Since the adversary is bounded by $poly(\lambda)$, we can assume the adversary requests all needed public keys at the start.

  Next, we may break the adversary into two algorithms.  The first receives the challenge public key, along with $c$ additional public keys.  This algorithm outputs two messages (perhaps randomly) and also a quantum state to be used by the second stage adversary algorithm.  This second stage algorithm will receive the challenge ciphertext and the quantum state output by the first algorithm (which may contain any and all unused quantum public keys from the initial stage).  It will then output a single bit.  

  Let's consider the first stage.  The adversary receives $r^*$ along with $\ket{\psi_{f,r^*}}$ where:
  \begin{equation}
      \ket{\psi_{f,r}} = \bigotimes_{g=1}^N\sum_{x\in\{0,1\}^n}(-1)^{f(r,g,x)}\ket{x}.
  \end{equation}
  It also receives $r_1, \cdots, r_c$ along with $\ket{\psi_{f,r_1}},\cdots,\ket{\psi_{f,r_c}}$.  The challenge public key (which is a copy of $\ket{\psi_{f,r^*}}$) is kept by the oracle, but evolves to $\mathcal{E}(\kb{\psi_{f,r^*}})$ which will create a noisy public key to be used later in encryption. The adversary eventually outputs a challenge message pair $m_0$ and $m_1$ of length $\ell$.

  Consider the state of the quantum system (the joint state modeling the adversary and the oracle) after the adversary outputs a message pair.  This state, denoted $\rho$, can be written:
  \begin{align}%
    \rho=\sum_fp(f)\sum_{r^*}p(r^*)\kb{\widetilde{\psi}_{f,r^*}}\sum_{\vec{r}}p(\vec{r})\sigma_{ME}(f, r^*, \vec{r})
  \end{align}
  where, above, the first sum is over all random functions $f$ the oracle might have chosen; the second is over all challenge public key tokens $r^*$ followed by the noisy challenge public key held by the oracle denoted $\kb{\widetilde{\psi}}$ (which may not be a pure state); and
  where the last sum, over $\vec{r}$, is over all possible public key tokens $r\up{1},\cdots, r\up{c}$ which are chosen by the public key generation function.  The state $\sigma_{ME}$ denotes the messages chosen by the adversary (a classical mixed state) and Eve's quantum ancilla which includes any classical and quantum state information she needs for her second stage attack, and also includes any, and all, quantum public keys.

  Now, consider a related state $\tau$ which will be the same quantum state as $\rho$, above, except where each $r\up{y} \ne r^*$ for all $y = 1, \cdots, c$.  That is, $\tau$ is the same experiment, however where all random tokens generated are different from the challenge one.  It is not difficult to show that:
  \begin{equation}
    \frac{1}{2}\trd{\rho-\tau} \le \Pr(r_i = r^*\text{ for at least one $i$}) \le \frac{c^2}{2^p} = \epsilon_R.
  \end{equation}
  Clearly, this is negligible in $p$; if $p = \eta \lambda$, it is negligible in $\lambda$.

  Let's consider $\tau$ now.  Since $r\up{y}\ne r^*$ for all public keys, it is clear that each additional requested public key is independent of the challenge key and so cannot provide Eve with any information on the challenge ciphertext.  Furthermore, and also because of this, the value of $f(r^*,\cdot, \cdot)$, used in the creation of the challenge public key, is independent of all other evaluations of $f$ for all other keys.  Thus, we may replace $f$ with, instead, a random bit string, similar to what is done in the RR-QKD protocol (see Section \ref{sec:RRQKD}).  Finally, since each group of the challenge cipher text (particularly, the noisy version held by the adversary) is independent, it can be shown, using Theorem \ref{thm:intro:RR-rate} along with basic properties of min entropy \cite{renner2008security}, that:
  \begin{equation}
    \Hmin(W'|IJE)_\tau \ge N(1-\log(1+2^{1-n})),
  \end{equation}
  where $W'$ is the computed pad value (see Section \ref{sec:protocol}), namely it is the random variable taking values $f(r^*, g, i) \oplus f(r^*,g,j)$, or, equivalently, taking values based on the secret bit string chosen initially as in RR-QKD.
  
  Let $W$ be the random variable storing the measured pad value (see Section \ref{sec:protocol}) which is produced from the actual measurement outcome on the noisy challenge public key.  Since this is what's actually used for encryption, we need a bound on $\Hmin(W'|IJE)$ in order to use privacy amplification (Equation \ref{eq:intro:PA}).  Due to our assumptions on the noise channel, in particular since $W$ can be derived from $W'$ by adding bit-flip errors independently on all bits and independently of the adversary, it holds that $\Hmin(W|IJE)_\tau \ge \Hmin(W'|IJE)_\tau$.    Of course, an error correcting code is added to the ciphertext which is also given to the adversary.  Let $C$ be the random variable storing the error correction information.  From the chain rule of min entropy \cite{renner2008security}, while also bounding the error correction leakage by $H(Q)$, where $Q$ is the bit flip probability of the noise map we consider (see, for instance, \cite{tomamichel2012tight}), we have:
  \begin{equation}
    \Hmin(W|IJEC)_\tau \ge N(1-\log_2(1+2^{1-n})-H(Q)),
  \end{equation}

  Now, since $\rho$ and $\sigma$ are $\epsilon_R$ close in trace distance, we switch to smooth quantum min entropy (see Section \ref{sec:background}) and we have:
  \begin{equation}
    \begin{split}
        \Hmin^{2\epsilon_R}(W|IJEC)_\rho &\ge \Hmin(W|IJEC)_\tau \\ 
                                         &\ge N(1-\log_2(1+2^{1-n})-H(Q))
    \end{split}
  \end{equation}
  Thus, if we have
  \begin{equation}
    \ell \le N(1-\log_2(1+2^{1-n})-H(Q)) - 2\log\frac{1}{\epsilon_{PA} - 4\epsilon_R}
  \end{equation}
  then it will hold, using Equation \ref{eq:intro:PA}, that the measured pad system $W$ after privacy amplification, namely $h(W)$, where $h$ is a randomly chosen two-universal hash function, is $\epsilon_{PA}$-close (for $\epsilon_{PA} > 4\epsilon$) to an ideal secret key of size $\ell$ bits that is independent of the adversary system (i.e., it is $\epsilon_{PA}$ close to a OTP secret key).  Of course, the probability of winning the experiment given a OTP key is only $1/2$.  Thus, the probability of winning the experiment given a random function $f$, namely $\prod^f$, is:
  \begin{equation}
    \Pr(\experiment_{\prod^f,\adv} = 1) \le \frac{1}{2} + \frac{1}{2}\epsilon_{PA}.
  \end{equation}
  If we set $\epsilon_{PA} = 5\epsilon_R$, which is negligible in $\lambda$ (if $p = \eta\lambda$), the result follows.


\end{proof}

The above theorem shows how many groups $N$ are needed in the public key in order to be $\cpanet$ secure.  We comment that $n$ may be small while still maintaining the PRF security guarantee (which, itself, depends on the key-size $\lambda$).  There is an interesting trade-off in terms of $N$ and $n$.  For example, Figure \ref{fig:security:fig1}, shows the needed number of groups $N$, for increasing $n$ while setting $p=256$ and $c=10^5$ with $Q=5\%$ or $10\%$.  As can be seen, as $n$ increases, the size of $N$ decreases meaning that, as the total number of qubits in superposition increases, the total number of needed groups decreases.  However, if we consider the actual number of qubits (namely $nN$), we see there tends to be an optimal value for $n$ if the goal is to minimize the total number of qubits as seen in Figure \ref{fig:security:fig2}.  Our Theorem~\ref{thm:security} can be used in practice to find an optimal setting for $n$ and $N$, which balances the total size of the public key ($nN$) with the maximal allowed qubits in superposition supported by the NISQ device ($n$).  Note that each group can be created independently; however within a group, the $n$ qubits must act coherently.  Thus, at least on current NISQ devices, it is best to have a small $n$ and, importantly, we see in Figure \ref{fig:security:fig2}, that our protocol takes its optimal value for small values of $n$ ranging from three to five (at least for these noise and message length settings).  Note that this is, to our knowledge, the first quantum PKE scheme using quantum secure PRFs that can be used, securely, on today's devices as only a small number of qubits in superposition are needed along with its support for classical error correction.

\begin{figure}
  \centering
  \includegraphics[width=1.0\linewidth]{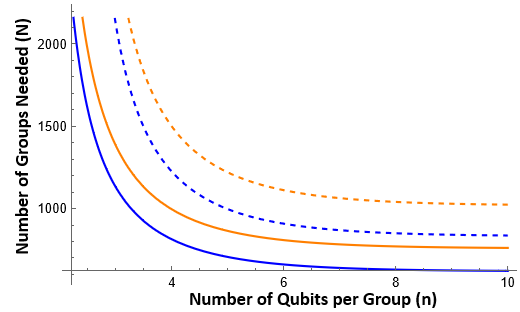}
  \caption{Showing the necessary setting for $N$ (number of groups in a single public key) according to our Theorem~\ref{thm:security} as the number of qubits in each group ($n$) increases.  Solid lines: Setting $Q=5\%$; Dashed Lines: Setting $Q = 10\%$.  Blue: Encrypting one bit ($\ell = 1$); Orange: Encrypting $\ell = 100$ bits.}\label{fig:security:fig1}
\end{figure}

\begin{figure}
  \centering
  \includegraphics[width=1.0\linewidth]{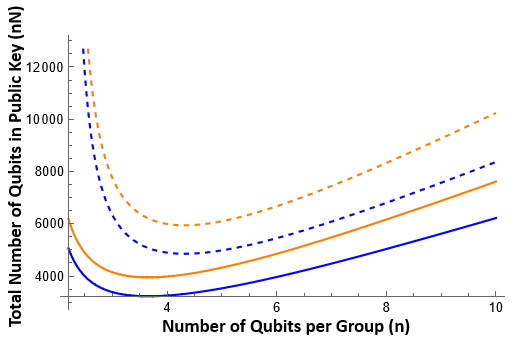}
  \caption{Similar to Figure \ref{fig:security:fig1}, but now showing the total number of qubits needed per public key ($nN$).  We see that there is an optimal setting for $n$ around three to five depending on the noise.  Note that while a large number of qubits are required, they do not have to all act together coherently; instead only $n$ qubits ($x$-axis) must act coherently.}\label{fig:security:fig2}
\end{figure}


\section{Design and Implementation} \label{sec:design}
We next present a concrete circuit design and implementation of the protocol discussed in Section~\ref{sec:protocol}. It is for the case $n=3$ (i.e., the dimension of the quantum state in each round is $2^3$) for both efficiency and practicality.  



For a $3$-qubit implementation of the protocol, our quantum PRF
family is then composed of functions $F_k(r,g,\cdot):\mathbb{F}_2^3 \to \mathbb{F}_2$. As the public key $\ket\psi = \frac{1}{\sqrt{8}} \sum_x (-1)^{F_k(r,g,x)} \ket{x}$ is unchanged (specifically changes by a global phase) if we replace $F_k(r,g,\cdot) \mapsto F_k(r,g,\cdot) \oplus 1$, two functions are equivalent if they differ by a constant. Therefore, there are only $128$ nonequivalent functions to choose from. We parameterize our functions as
\begin{equation}
    \begin{split}
        F_\ell(x)  &= \ell_0 x_0 +  \ell_1 x_1 +  \ell_2 x_2 +  \ell_3 x_0 x_1 \\
                &\quad +  \ell_4 x_0 x_2 +  \ell_5 x_1 x_2 +  \ell_6 x_0 x_1 x_2, 
    \end{split}
\end{equation}
where each $\ell_j = \ell_j(k,r,g)$. For modeling purposes, it suffices to select a function by sampling a ``key'' $\ell \stackrel{\$}{\leftarrow} \mathbb{F}_2^7$. We show the circuit for implementing $F_\ell$ in (\ref{eqn:public-key-circuit}) below, where each boxed operation is only applied if the associated key bit ($\ell_j$) is set. That is, each operation is classically controlled on the labeled key bit.
\begin{equation}\label{eqn:public-key-circuit}
    \prfcircuit
\end{equation}

\begin{subsection}{Implementing Quantum PKE Encryption} \label{subsection:implementing_encryption}
    In the following,  we first treat the general case and then turn to the case when $n = 3$. Our public key (omitting normalization) is of the form
    \[
        \ket{\psi_{k, r}} = \bigotimes_{g=1}^{N} \sum_{x \in \{0, 1\}^n} { (-1)^{\prf{x}} \ket{x} }.
    \]
    We then draw $\delta$ uniformly from the set $\{1, \cdots, 2^{n} - 1\}$. Once $\delta$ is fixed, the set $\{0, \cdots, 2^{n}-1\}$ has $2^{n-1}$ distinct pairs $(i_t, j_t)$ such that $j_t - i_t = \delta \pmod{2^n}$, for $t \in \{0, \cdots, 2^{n-1}-1\}$. The process of obtaining the bit $w_g$ for each group $g$ in our public key is similar and consists of three steps:
    \begin{enumerate}
        \item Choose an index $b \in \{0, \cdots, n-1\}$. Apply the unitary $U_\pi(b, \delta)$ that implements the permutation $\pi$ that we describe below.
        \item Apply the Hadamard gate to qubit ``$b$". 
        \item Measure in the computational basis to obtain value $\pi(i_t)$ or $\pi(j_t)$ for some $t$. This determines the value of $w_g$. Use the inverse permutation $\pi^{-1}$ to obtain $(i_t, j_t)$.
    \end{enumerate}

    We begin with the first step. The index $b$ is chosen from $\{0, \cdots, n-1\}$. Let $\mathsf{int}(x)$ be the integer representation of the $n$-bit binary number $x=x_0\cdots x_{n-1}$. Consider the permutation $\pi$ that maps pairs $(i_t, j_t)$ in the following way:
    \begin{align} \label{eqn:permutation_binary}
        \begin{split}
            i_t \rightarrow \mathsf{int}(x_0 \cdots x_{b-1} 0 x_{b+1} \cdots x_{n-1})\\
            j_t \rightarrow \mathsf{int}(x_0 \cdots x_{b-1} 1 x_{b+1} \cdots x_{n-1})
        \end{split}
    \end{align}
    i.e., it maps every $(i_t, j_t)$ pair to unique integers with $0$ and $1$ respectively in bit position $b$. Let $U_{\pi}(b, \delta) $ be the unitary that implements this permutation. Using this notation, the public key for a fixed group $g$ after applying the permutation unitary can be written as $ U_\pi(b, \delta) \ket{\psi_{k, g}}=$
    \begin{align*}
        \sum_{t}
                \Big(
                    (-1)^{\prf{i_t}} &\ket{x_0, \cdots, x_{a-1}, 0, x_{a+1}, \cdots, x_{n-1}} \\
                +   (-1)^{\prf{j_t}} &\ket{x_0, \cdots, x_{a-1}, 1, x_{a+1}, \cdots, x_{n-1}}
                \Big)
    \end{align*}
    Then as per step 2,
    applying the Hadamard on qubit $b$ transforms this state into
    \begin{align*}
            \sum_{t} 
                \Big(
                    \left[(-1)^{\prf{i_t}} + (-1)^{\prf{j_t}}\right] &\ket{\pi(i_t)} \\
            +       \left[(-1)^{\prf{i_t}} - (-1)^{\prf{j_t}}\right] &\ket{\pi(j_t)}
                \Big)
    \end{align*}
    The third step is to measure the state in the computational basis. Note that if $\prf{i_t} \oplus \prf{j_t} = 0$, $\left|(-1)^{\prf{i_t}} + (-1)^{\prf{j_t}}\right| = 2$ and the measurement yields outcome $\pi(i_t)$ with probability $1/2^{n-1}$. Similarly if $\prf{i_t} \oplus \prf{j_t} = 1$, the outcome is $\pi(j_t)$ with probability $1/2^{n-1}$. Whether we measure $\pi(i_t)$ or $\pi(j_t)$ determines the value of our secret bit $w_g = \prf{i_t} \oplus \prf{j_t}$, and the corresponding $(i_t, j_t)$ values can be obtained by computing the inverse permutation $\pi^{-1}$. 

    There are many choices for the permutation $\pi$ once $b$ is fixed, as is evident from (\ref{eqn:permutation_binary}). There are $2^{n-1}$ different pairs that $(i_0, j_0)$ can map to. Once the mapping for $(i_0, j_0)$ is fixed, there are $2^{n - 1} - 1$ choices for $(i_1, j_1)$ and so on. Furthermore, within each pair, swapping the $i_t$ and $j_t$ does not make a difference. This results in $2^{n}!/2$ different permutations that can be used to pick a single pair $(i, j)$ uniformly at random with the property that $j - i = \delta \pmod{2^n}$.
    
    For $n=3$, $\delta$ is chosen uniformly from the set $\{1, \cdots, 7\}$. The case when $\delta=1, 2$ and $4$ is quite simple. We take $U_\pi(b,\delta) = \iden$ and take $b=0, 1$ and $2$ respectively. The other cases are non-trivial and we use the permutations given in (\ref{eqn:permutation_matrices}) below. For those, we take $b=0$ and apply the Hadamard gate to qubit $0$. This has the nice property that every $(i_t, j_t)$ pair is mapped to two consecutive integers and the inverse permutation mapping back to the $(i_t, j_t)$ pair follows nicely based on whether the measured value is even or odd. Concretely, for measurement outcome $x$, if $x$ is even, then:
        $(i, j)_g = (\pi^{-1}(x), \pi^{-1}(x + 1))$ and $w_g = \prf{i} \oplus \prf{j} = 0.$
    Otherwise, if $x$ is odd:
        $(i, j)_g = (\pi^{-1}(x - 1), \pi^{-1}(x))$, and $w_g = \prf{i} \oplus \prf{j} = 1.$
    

    \begin{align} 
        \begin{split} \label{eqn:permutation_matrices}
            \delta = 3: &   \begin{pmatrix}
                             0 & 3 & 1 & 6 & 2 & 5 & 4 & 7 \\
                             0 & 1 & 2 & 3 & 4 & 5 & 6 & 7
                        \end{pmatrix} \\
            \delta = 5: &   \begin{pmatrix}
                             0 & 5 & 1 & 4 & 2 & 7 & 3 & 6 \\
                             0 & 1 & 2 & 3 & 4 & 5 & 6 & 7
                        \end{pmatrix} \\
            \delta = 6: &   \begin{pmatrix}
                             0 & 6 & 1 & 7 & 2 & 4 & 3 & 5 \\
                             0 & 1 & 2 & 3 & 4 & 5 & 6 & 7 
                        \end{pmatrix} \\
            \delta = 7: &   \begin{pmatrix}
                             0 & 7 & 1 & 2 & 3 & 4 & 5 & 6 \\
                             0 & 1 & 2 & 3 & 4 & 5 & 6 & 7
                        \end{pmatrix}
        \end{split}
    \end{align}
    
    The overall encryption circuits corresponding to each $\delta$ in (\ref{eqn:permutation_matrices}) are shown in \figurename~\ref{fig:circ_all_deltas}.


    To summarize, the overall circuit (corresponding to the protocol in Figure~\ref{fig:protocol_enc_dec}) is as follows. Alice performs the circuit described in (\ref{eqn:public-key-circuit}) to generate the public key.  Bob picks a random $\delta$ and performs the corresponding encryption unitary based on $\delta$ (as described earlier, simple circuits when $\delta=0$, 1, or 2, and see Figure~\ref{fig:circ_all_deltas} for $\delta=3$, 5, 6 or 7) to obtain secret bits
    $w = w_0\cdots w_{n-1}$. After that, 
    Bob computes the ciphertext classically by first sampling a hash function $h$ from a two universal hash family, and then computing $h(w) \oplus m$, where $m$ is the message we want to encrypt.

    \begin{figure}[!t]
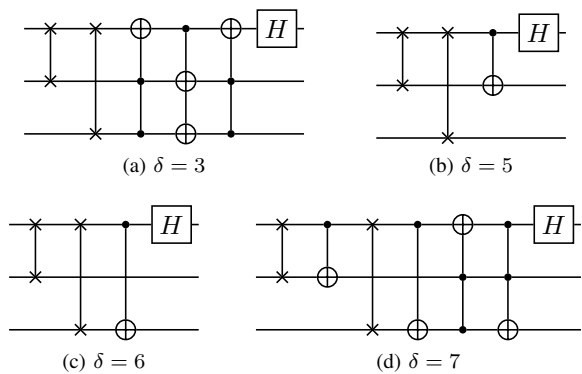

        \centering
        \subfloat[$\delta = 3$]{\dthree \label{fig:circ_delta_3}} \hfil
        \subfloat[$\delta = 5$]{\dfive  \label{fig:circ_delta_5}} \hfil
        \subfloat[$\delta = 6$]{\dsix   \label{fig:circ_delta_6}} \hfil
        \subfloat[$\delta = 7$]{\dseven \label{fig:circ_delta_7}}
        \caption{Encryption circuits for different values of $\delta=3$, 5, 6 or 7. The circuits implement $U_\pi(0,\delta)$ followed by a Hadamard gate on qubit $0$.}
        \label{fig:circ_all_deltas}
    \end{figure}
    
\end{subsection}

\section{Evaluation} \label{sec:eval}

We tested our implementation on Quantinuum's H1-1 trapped-ion device using Microsoft Azure Quantum \cite{quantinuum2025, azure_quantum}.  For simplicity, we considered a very basic function for our PRF of $F_k(x) = x_1$.  Constructing a circuit to implement an actual PRF $F_k(r,g,\cdot)$, for some $k$, $r$, and $g$, can be done by constructing a suitable circuit based on a look-up table given the output of a quantum secure PRF.  For our evaluation, we use a simple function just to test our measurement circuit implementation.  Our software is available online \cite{sourcecode}.

Our implementation allows for the creation of $G$ groups per trial run (thus using $4G$ qubits total since $n=3$ and we require one additional qubit per group to implement the PRF oracle); we tested $G=1$, $G=3$ (resulting in $12$ qubits used total) and $G=4$ (resulting in $16$ qubits total).  On each trial, we created $G$ public key groups, performed a measurement using the process described in Section~\ref{sec:design}, choosing different $\delta$ values for each of the $G$ groups, and repeated this process until we had at least $N=100$ total groups and thus our bit strings $W$ and $W'$ were at least 100 bits long.  We then computed the error in the bit strings (with $W'$ being the resulting string from the PRF function, and $W$ being the observed outcome).  The total observed noise is shown in Table~\ref{tab:eval:H1-1} along with the needed number of groups for security $N$ according to our Theorem~\ref{thm:security}.  Note that the noise increases with $G$, which is consistent with 
the interconnect structure of the Quantinuum computer.

\begin{table}
\caption{Average noise on the Quantinuum H1-1 trapped-ion device.  Also computing the minimum number of groups needed to encrypt eight bit messages, using our Theorem~\ref{thm:security} and the observed bit error rate, assuming $p = 128$ and $c = 10^5$.}\label{tab:eval:H1-1}
  \begin{tabular}{c|c|l}
    Number of groups per trial & Average error & Required Size $N$ \\
    \hline
    $G = 1$ & $1.65\%$ Error & $N=794$ Groups\\
    $G = 3$ & $8.89\%$ Error & $N=1802$ Groups\\
    $G = 4$ & $11.54\%$ Error & $N=2727$ Groups
  \end{tabular}
\end{table}

Looking forward to larger implementations, we can also calculate the expected error rate of our circuits.  As our circuits are small, we can perform an exact error analysis, using the methodology of \cite{kukliansky2025quantum}. In this formalism each classical or quantum operation becomes ``circuit enumerator'' whose variables symbolically track the effect of errors. For simplicity we use the so-called circuit error model, where:
    (1) each unitary operation (including the identity when a qubit is idling) suffers a post-applied random Pauli error with a probability $\epsilon$; and
    (2) each measurement suffers from a readout assignment error (bit-flip error) with probability $\rho$.
We assume classical operations are error-free.

The result is a tensor representation of the circuit that includes all possible ways that errors could affect the result. From this we can model error statistics; for our purpose we compute the likelihood of a decryption failure as a function of $\ell$, $\delta$, $\epsilon$, and $\rho$. As $\ell$ is well-modeled as uniformly distributed over its $128$ possible values, we can compute the expectation of a decryption failure for each $\delta$ as a function of the error rates $\epsilon$ and $\rho$. The truncated results are shown in Table~\ref{table:error-modeling} which can be used to estimate the error at larger scales.

 \begin{table}[!t]
\caption{Expected probability of decryption failure.}\label{table:error-modeling}
 \begin{center}
    \renewcommand{\arraystretch}{1.4}
    \begin{tabular}{c|l}
    $\delta$ & Probability of decryption failure\\\hline
    $1$ & $3.5\epsilon + 2\rho - 10.8\epsilon^2 - 14.4\epsilon\rho - 2.5\rho^2 + O(\epsilon,\rho)^3$\\
    $2$ & $3.4\epsilon + 2\rho - 10.8\epsilon^2 - 14.4\epsilon\rho - 2.5\rho^2 + O(\epsilon,\rho)^3$\\
    $3$ & $8.0\epsilon + 2\rho - 60.4\epsilon^2 - 31.9\epsilon\rho - 2.5\rho^2 + O(\epsilon,\rho)^3$\\
    $4$ & $3.3\epsilon + 2\rho -  9.3\epsilon^2 - 13.6\epsilon\rho - 2.5\rho^2 + O(\epsilon,\rho)^3$\\
    $5$ & $5.8\epsilon + 2\rho - 30.7\epsilon^2 - 23.2\epsilon\rho - 2.5\rho^2 + O(\epsilon,\rho)^3$\\
    $6$ & $5.8\epsilon + 2\rho - 30.5\epsilon^2 - 23.3\epsilon\rho - 2.5\rho^2 + O(\epsilon,\rho)^3$\\
    $7$ & $8.0\epsilon + 2\rho - 60.4\epsilon^2 - 31.9\epsilon\rho - 2.5\rho^2 + O(\epsilon,\rho)^3$\\
    \end{tabular}\end{center}
\end{table}


\section{Conclusion}
In this paper, 
we have designed a practical quantum PKE
scheme that accounts for the constraints of current NISQ devices. Specifically, our design only requires a small number of qubits to be in superposition, and is tolerant to the noise of current NISQ devices.
We have designed and
implemented circuits for our PKE scheme and tested them on current quantum computers,
evaluating the error rate and efficiency of the schemes.   Many interesting future problems remain.  Our design, based on the RR-QKD protocol, may serve as a template for other quantum cryptographic primitives beyond encryption.  It would also be interesting to try and extend our construction to be CCA secure as defined in \cite{Barooti24:qPKE} for quantum PKE.  Finally, it may be interesting to extend our security model, which currently takes into account natural noise, to the case of adversarial noise.  For instance, we may model public key distribution over a noisy network, where the network is controlled by the adversary. We suspect our construction would be secure in this case, though an exact security definition, and proof, we leave as future work.

\balance
\bibliographystyle{IEEEtran}
\bibliography{IEEEabrv,all,walter}

\begin{thebibliography}{10}
\providecommand{\url}[1]{#1}
\csname url@samestyle\endcsname
\providecommand{\newblock}{\relax}
\providecommand{\bibinfo}[2]{#2}
\providecommand{\BIBentrySTDinterwordspacing}{\spaceskip=0pt\relax}
\providecommand{\BIBentryALTinterwordstretchfactor}{4}
\providecommand{\BIBentryALTinterwordspacing}{\spaceskip=\fontdimen2\font plus
\BIBentryALTinterwordstretchfactor\fontdimen3\font minus
  \fontdimen4\font\relax}
\providecommand{\BIBforeignlanguage}[2]{{%
\expandafter\ifx\csname l@#1\endcsname\relax
\typeout{** WARNING: IEEEtran.bst: No hyphenation pattern has been}%
\typeout{** loaded for the language `#1'. Using the pattern for}%
\typeout{** the default language instead.}%
\else
\language=\csname l@#1\endcsname
\fi
#2}}
\providecommand{\BIBdecl}{\relax}
\BIBdecl

\bibitem{Morimae2022}
\BIBentryALTinterwordspacing
T.~Morimae and T.~Yamakawa, ``One-wayness in quantum cryptography,''
  \emph{arXiv preprint arXiv:2210.03394}, 2022. [Online]. Available:
  \url{https://arxiv.org/abs/2210.03394}
\BIBentrySTDinterwordspacing

\bibitem{Barooti2023}
\BIBentryALTinterwordspacing
K.~Barooti, G.~Malavolta, and M.~Walter, ``A simple construction of quantum
  public-key encryption from quantum-secure one-way functions,''
  \emph{Cryptology ePrint Archive}, vol. Paper 2023/306, 2023. [Online].
  Available: \url{https://eprint.iacr.org/2023/306}
\BIBentrySTDinterwordspacing

\bibitem{Grilo2023}
\BIBentryALTinterwordspacing
A.~B. Grilo, O.~Sattath, and Q.-H. Vu, ``Encryption with quantum public keys,''
  \emph{Cryptology ePrint Archive}, vol. Paper 2023/345, 2023. [Online].
  Available: \url{https://eprint.iacr.org/2023/345}
\BIBentrySTDinterwordspacing

\bibitem{Barooti24:qPKE}
K.~Barooti, A.~B. Grilo, L.~Huguenin-Dumittan, G.~Malavolta, O.~Sattath, Q.-H.
  Vu, and M.~Walter, ``Public-key encryption with quantum keys,'' in
  \emph{Theory of Cryptography Conference}.\hskip 1em plus 0.5em minus
  0.4em\relax Springer, 2023, pp. 198--227.

\bibitem{coladangelo2023quantum}
A.~Coladangelo, ``Quantum trapdoor functions from classical one-way
  functions,'' \emph{arXiv preprint arXiv:2302.12821}, 2023.

\bibitem{Impagliazzo1995}
R.~Impagliazzo, ``A personal view of average-case complexity,'' in \emph{Proc.
  of Structure in Complexity Theory}, 1995.

\bibitem{nikolopoulos2008applications}
G.~M. Nikolopoulos, ``Applications of single-qubit rotations in quantum
  public-key cryptography,'' \emph{Physical Review A—Atomic, Molecular, and
  Optical Physics}, vol.~77, no.~3, p. 032348, 2008.

\bibitem{vlachou2015quantum}
C.~Vlachou, J.~Rodrigues, P.~Mateus, N.~Paunkovi{\'c}, and A.~Souto, ``Quantum
  walk public-key cryptographic system,'' \emph{International Journal of
  Quantum Information}, vol.~13, no.~07, p. 1550050, 2015.

\bibitem{ananth2022cryptography}
P.~Ananth, L.~Qian, and H.~Yuen, ``Cryptography from pseudorandom quantum
  states,'' in \emph{Annual International Cryptology Conference}.\hskip 1em
  plus 0.5em minus 0.4em\relax Springer, 2022, pp. 208--236.

\bibitem{Okamoto2000}
T.~Okamoto, K.~Tanaka, and S.~Uchiyama, ``Quantum public-key cryptosystems,''
  in \emph{CRYPTO 2000}, ser. Lecture Notes in Computer Science, M.~Bellare,
  Ed., vol. 1880.\hskip 1em plus 0.5em minus 0.4em\relax Springer, Heidelberg,
  August 2000, pp. 147--165.

\bibitem{Gottesman2005}
\BIBentryALTinterwordspacing
D.~Gottesman, ``Quantum public key cryptography with information-theoretic
  security,'' Presentation, 2005. [Online]. Available:
  \url{https://www2.perimeterinstitute.ca/personal/dgottesman/Public-key.ppt}
\BIBentrySTDinterwordspacing

\bibitem{Kawachi2005}
A.~Kawachi, T.~Koshiba, H.~Nishimura, and T.~Yamakami, ``Computational
  indistinguishability between quantum states and its cryptographic
  application,'' in \emph{EUROCRYPT 2005}, ser. Lecture Notes in Computer
  Science, R.~Cramer, Ed., vol. 3494.\hskip 1em plus 0.5em minus 0.4em\relax
  Springer, Heidelberg, May 2005, pp. 268--284.

\bibitem{sasaki2014practical}
T.~Sasaki, Y.~Yamamoto, and M.~Koashi, ``Practical quantum key distribution
  protocol without monitoring signal disturbance,'' \emph{Nature}, vol. 509,
  no. 7501, pp. 475--478, 2014.

\bibitem{sourcecode}
``Source code,'' \url{https://walterkrawec.org/code/qpke-test.py}.

\bibitem{renner2008security}
R.~Renner, ``Security of quantum key distribution,'' \emph{International
  Journal of Quantum Information}, vol.~6, no.~01, pp. 1--127, 2008.

\bibitem{konig2009operational}
R.~Konig, R.~Renner, and C.~Schaffner, ``The operational meaning of min-and
  max-entropy,'' \emph{IEEE Transactions on Information theory}, vol.~55,
  no.~9, pp. 4337--4347, 2009.

\bibitem{goldreich1986construct}
O.~Goldreich, S.~Goldwasser, and S.~Micali, ``How to construct random
  functions,'' \emph{Journal of the ACM (JACM)}, vol.~33, no.~4, pp. 792--807,
  1986.

\bibitem{zhandry2021construct}
M.~Zhandry, ``How to construct quantum random functions,'' \emph{Journal of the
  ACM (JACM)}, vol.~68, no.~5, pp. 1--43, 2021.

\bibitem{ananth2022pseudorandom}
P.~Ananth, A.~Gulati, L.~Qian, and H.~Yuen, ``Pseudorandom (function-like)
  quantum state generators: New definitions and applications,'' in \emph{Theory
  of Cryptography Conference}.\hskip 1em plus 0.5em minus 0.4em\relax Springer,
  2022, pp. 237--265.

\bibitem{skoric2017short}
B.~Skoric, ``A short note on the security of round-robin differential
  phase-shift qkd,'' \emph{IACR Cryptology ePrint Archive}, 2017.

\bibitem{Elkouss2009:LDPC}
D.~Elkouss, J.~Martinez-Mateo, and V.~Martín, ``Information reconciliation for
  quantum key distribution,'' in \emph{2009 IEEE International Symposium on
  Information Theory (ISIT)}.\hskip 1em plus 0.5em minus 0.4em\relax IEEE,
  2009, pp. 1879--1883.

\bibitem{MinkNakassis2012:LDPC}
\BIBentryALTinterwordspacing
A.~Mink and A.~Nakassis, ``{LDPC for QKD Reconciliation},'' 2012. [Online].
  Available: \url{https://arxiv.org/abs/1205.4977}
\BIBentrySTDinterwordspacing

\bibitem{BrassardSalvail1994:Cascade}
G.~Brassard and L.~Salvail, ``Secret-key reconciliation by public discussion,''
  in \emph{Advances in Cryptology -- EUROCRYPT '93}, ser. Lecture Notes in
  Computer Science, T.~Helleseth, Ed., vol. 765.\hskip 1em plus 0.5em minus
  0.4em\relax Berlin, Heidelberg: Springer, 1994, pp. 410--423.

\bibitem{tomamichel2012tight}
M.~Tomamichel, C.~C.~W. Lim, N.~Gisin, and R.~Renner, ``Tight finite-key
  analysis for quantum cryptography,'' \emph{Nature communications}, vol.~3,
  no.~1, p. 634, 2012.

\bibitem{quantinuum2025}
Quantinuum H1-1, \url{https://www.quantinuum.com}, April 7-11, 2025.

\bibitem{azure_quantum}
{Azure Quantum}, \url{https://learn.microsoft.com/en-us/azure/quantum/}, 2025.

\bibitem{kukliansky2025quantum}
A.~Kukliansky and B.~Lackey, ``Quantum circuit tensors and enumerators with
  applications to quantum fault tolerance,'' \emph{IEEE Transactions on
  Information Theory}, 2025.

\end{thebibliography}

\end{document}